\documentclass[english,11pt,a4paper]{amsart}
\usepackage{amstext,amsthm,amsfonts,amsmath}
\usepackage{graphicx}
\usepackage{babel}
\usepackage{color} 
\usepackage{amsaddr}
\usepackage{hyperref}

\makeatletter

\newcommand{\D}{\mathcal{D}}

\newcommand{\A}{\mathcal{A}}

\renewcommand{\H}{\mathcal{H}}
\renewcommand{\L}{\mathcal{L}}
\newcommand{\N}{\mathbb{N}}

\newcommand{\R}{\mathbb{R}}
\newcommand{\C}{\mathbb{C}}

\newcommand{\unit}{\mathbf{1}}
\newcommand{\ra}{\rightarrow}

\newcommand{\U}{\operatorname{U}}
\newcommand{\tr}{\operatorname{tr }}

\newcommand{\rmd}{\operatorname{d\!}}
\newcommand{\rmi}{\operatorname{i}}
\newcommand{\rme}{\operatorname{e}}

\newcommand{\slim}{\operatorname*{s-lim }}
\renewcommand{\d}{\delta}

\renewcommand{\t}{\otimes}

\newcommand{\spa}{\operatorname{span }}
\newcommand{\Sch}{\mathcal{S}}
\newcommand{\boxy}{\hfill $\Box$\vspace{5mm}}

\renewcommand{\lg}{\mathfrak{g}}

\def\d{\mathrm{d}}

\newtheorem{theorem}{Theorem}
\theoremstyle{definition}\newtheorem{definition}[theorem]{Definition}
\theoremstyle{definition}\newtheorem{remark}[theorem]{Remark}
\theoremstyle{definition}\newtheorem{example}[theorem]{Example}
\theoremstyle{definition}\newtheorem{conjecture}[theorem]{Conjecture}
\theoremstyle{definition}\newtheorem{problem}[theorem]{Problem}
\numberwithin{theorem}{section}

\makeatother

\title{Dynamical Decoupling of Unbounded Hamiltonians}
\author{Christian Arenz}
\address{Frick Laboratory, Princeton University, Princeton NJ 08544, USA}
\author{Daniel Burgarth}
\address{Department of Mathematics, Aberystwyth University, Aberystwyth SY23 3BZ, UK}
\author{Paolo Facchi}
\address{Dipartimento di Fisica and MECENAS, Universit\`a di Bari, I-70126 Bari, Italy}
\address{INFN, Sezione di Bari, I-70126 Bari, Italy}
\author{Robin Hillier}
\address{Department of Mathematics and Statistics, Lancaster University, Lancaster LA1 4YF, UK}

\begin{document}
\begin{abstract}
We investigate the possibility to suppress interactions between a finite dimensional system and an infinite dimensional environment through a fast sequence of unitary kicks on the finite dimensional system. This method, called dynamical decoupling, is known to work for bounded interactions, but physical environments such as bosonic heat baths are usually modelled with unbounded interactions, whence here we initiate a systematic study of dynamical decoupling for unbounded operators. We develop a sufficient decoupling criterion for arbitrary Hamiltonians and a necessary decoupling criterion for semibounded Hamiltonians. We give examples for unbounded Hamiltonians where decoupling works and the limiting evolution as well as the convergence speed can be explicitly computed. We show that decoupling does not always work for unbounded interactions and provide both physically and mathematically motivated examples.
\end{abstract}
\maketitle
\tableofcontents{}
\thispagestyle{empty}

\section{Introduction and overview}

A powerful strategy to protect a quantum system from decoherence is dynamical decoupling~\cite{LB13}. The application of frequent and instantaneous unitary operations (kicks), which correspond to strong classical pulses applied to the system, makes it possible to average the system-environment interactions to zero. Originally dynamical decoupling dates back to pioneering work of Haeberlen and Waugh~\cite{HW68, WHH68}, who developed pulse sequences, such as spin-echo techniques, in order to increase the resolution in nuclear magnetic resonance. Later, these schemes were generalized by Viola and Lloyd~\cite{VKL98,VKL99, VK05}, establishing a theoretical framework that allows to suppress generic system-environment interactions.  Its particular strength is that it is applicable even if the details of the system-environment coupling are unknown.

Since perfect decoupling only happens in the limit of infinitely frequent kicks, in practice it is important to understand the convergence speed. In finite dimensions, error estimates are given in terms of the higher orders of the Magnus expansion or the Dyson series~\cite{LB13,VK05}. Here the existence of and the speed of convergence to the decoupled dynamics relies on norm bounds of the Hamiltonian~\cite{KL08}, allowing one to prove that dynamical decoupling works arbitrarily well on a finite time scale.   

However, real physical environments, such as the free electromagnetic field, are (to a good approximation) infinite dimensional. In particular, the description of system-environment interactions through potentially unbounded operators makes it challenging to decide whether dynamical decoupling works and, moreover, estimate the time-scales necessary to efficiently dynamically decouple the system from the environment. Series expansions are a touchy business \cite{Reed} and norm bounds diverge. The main purpose of this paper is to establish criteria and examples for  dynamical decoupling  of unbounded Hamiltonians.

From a physical perspective, dynamical decoupling  has to be faster than the fastest timescale of the overall dynamics~\cite{VKL99}, and it is typically argued 
that dynamical decoupling only works for environments yielding non-exponential decay~\cite{LB13}. It is argued that a `Zeno' region of
non-exponential decay (Fig.~\ref{fig:Exponential-and-non-exponential})
determines the time-scale for dynamical decoupling. 
However, this is a heuristic argument rather than a rigorous mathematical conclusion and we will provide several counterexamples to it below.
In fact, it is interesting to note that to decide whether dynamical decoupling works for infinite dimensional environments, the full Hamiltonian must be provided. That is, the reduced dynamics does not provide enough information, and for the same reduced dynamics there can be dilations (given by system-environment Hamiltonians and environment initial states) which can be decoupled, whereas others cannot. An example is given by qubit dephasing, for which the shallow pocket model~\cite{AHFB15} provides a dilation which can be decoupled, whereas its Cheborev-Gregoratti dilation~\cite{Wal} was recently shown to be not amenable to decoupling~\cite{GN17}. These two dilations can be considered as two extreme cases: the former being highly non-Markovian and the latter very singular with built-in Markovian properties. The true physical models are likely to be found in between such extremes, and it is important to find general criteria for decoupling.

\begin{figure}
\centering{}\includegraphics[width=0.8\columnwidth]{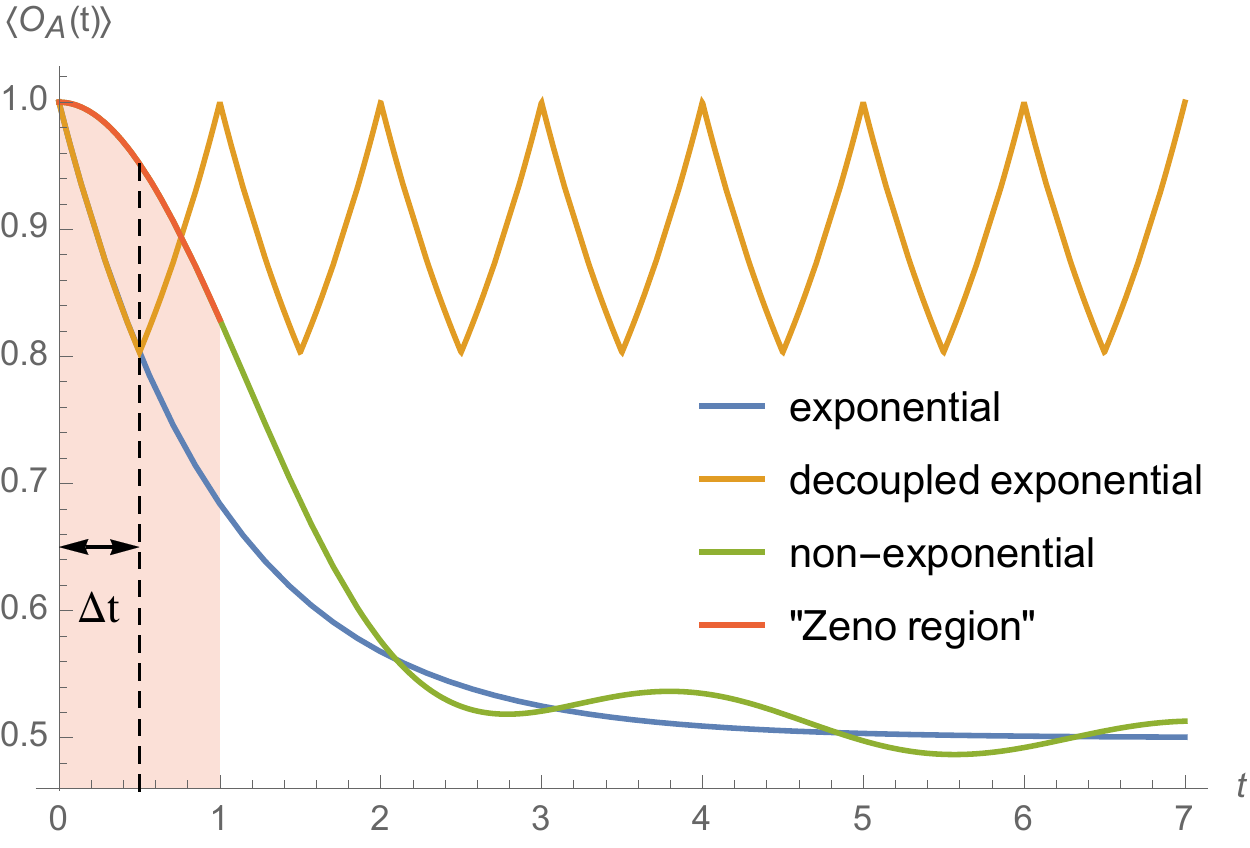}\caption{\label{fig:Exponential-and-non-exponential}Exponential and non-exponential
decay of some expectation $\left\langle O_{A}(t)\right\rangle $ and
the quadratic `Zeno region'. In the case of the shallow pocket model, the observable
shown is $p_{+}(t)/p_{+}(0)$ with $p_{+}(t)=\text{tr}\left\{ |+\rangle_{A}\langle+|\Lambda_{t}(\rho_{A})\right\} ,$
and the dynamics in presence of decoupling pulses with pulse time steps $\Delta t=0.5$
is shown. The non-exponential curve is an approximation of the shallow
pocket model with cut-off parameter $|x|\le 2$ in the Cauchy distribution. See Ex.~\ref{ex:sp} for details.}
\end{figure}

  In Section~\ref{sec:suff}, based on Trotter's product formula, we give a sufficient criterion for dynamical decoupling in Thm.~\ref{th:suff}, generalizing~\cite{AHFB15}. As an example we discuss the shallow pocket model~\cite{AHFB15},  which yields exponential decay but can be decoupled on arbitrary time-scales. Then we provide several generalizations which can be decoupled, but for which the time-scale of decoupling is non-trivial. Here we explicitly provide the corresponding time-scales in order to dynamically decouple the system from the environment and show that the efficiency depends on the initial bath state. Finally we provide an example showing that Thm.~\ref{th:suff} is sufficient, but not necessary for successful decoupling.  

In Section~\ref{dec:nec} we discuss lower bounded Hamiltonians, for which more can be
said about the convergence of the Trotter limit. Thm.~\ref{th:nec} provides
a necessary condition for dynamical decoupling of such Hamiltonians.
This is physically relevant as most reasonable interaction Hamiltonians
are unbounded above but bounded below. We provide an abstract example of a Hamiltonian where dynamical decoupling does not work. Finally, in Section~\ref{ex:Lee} we provide a generalization of the  Friedrichs-Lee
model which gives rise to an amplitude damping channel. We find that this
model cannot be dynamically decoupled and provide a physical interpretation.

\section{Prerequisites}\label{sec:prereq}

Consider a quantum mechanical system which is coupled to an environment. We suppose the system Hilbert space $\H_s\simeq \C^d$ to be $d$-dimensional (with $d$ finite) and the  environment Hilbert space $\H_e$ infinite-dimensional and separable.
We write $\H=\H_s\otimes\H_e$ for the total Hilbert space and $H$ for the total Hamiltonian, a self-adjoint operator with domain $\D(H)\subset \H$.

We assume that the initial state is uncorrelated $\rho=\rho_{s}\otimes\rho_{e}$,
where $\rho_s$ and $\rho_{e}$ are non-negative trace-class operators with $\tr\rho_s=\tr\rho_e=1$. With $U(t)=\rme^{-\rmi t H}$ we
refer to 
\begin{equation}
\label{eq:dynamical_map}
\rho_s \mapsto \Lambda_{t}(\rho_{s}):=\textrm{tr}_{\H_e}\left\{ U(t)(\rho_{s}\otimes\rho_{e})U^{*}(t)\right\}
\end{equation}
as the  dynamical map on state $\rho_s$ at time $t$. Here $\textrm{tr}_{\H_e}$ denotes the partial trace with respect to the factor $\H_e$.
The typical feature of reduced dynamics is
that certain expectation values 
$$\left<O_{s}(t)\right>:=\textrm{tr}\left\{ O_{s}\Lambda_{t}(\rho_{s})\right\} 
$$
of observables $O_{s}\in B(\H_s)$, the space of (bounded) linear operators on $\H_s$, can \emph{decay irreversibly}. This is particularly
so if the dynamical map $\Lambda_{t}$ has the semigroup property
$\Lambda_{t_1}\circ\Lambda_{t_2}=\Lambda_{t_1+t_2}$, for all $t_1,t_2 \geq 0$. In such a case, the reduced
dynamics $\rho_s(t) := \Lambda_{t}(\rho_{s})$ is described by a GKLS master equation $\dot{\rho}_s(t)=\mathcal{L}\rho_s(t)$
with generator 
\begin{equation}
\mathcal{L}(\rho_s)=- \rmi [H_s,\rho_s]+\sum_{i=1}^{d^{2}}\gamma_{i}\Big(L_{i}\rho_s L_{i}^{*}-\frac{1}{2}L_{i}^{*}L_{i}\rho_s-\frac{1}{2}\rho_s L_{i}^{*}L_{i}\Big),
\end{equation}
for all system densitity matrices $\rho_s$, where $\gamma_{i}\ge0$ and $L_{i}, H_s\in B(\H_s)$ and $H_s$ moreover self-adjoint.

\begin{definition}\label{def:dd}
A \emph{decoupling set for $\H_s$} is a finite group of unitary operators $V\subset \U(\H_s)$ such that
\[
\frac{1}{|V|} \sum_{v\in V} v x v^* = \frac{1}{d} \tr(x) \unit_{\H_s},\quad \text{for all } x\in B(\H_s).
\]
Let $N$ be a multiple of the cardinality $|V|$. A \emph{decoupling cycle} of length $N$ is a cycle $(v_1, v_2,\ldots,v_N)$  through $V$, that reaches each element of $V$ the same number of times.
\end{definition}
In~\cite{ABH17} it is shown that such a decoupling set always exists but it is usually not unique. Obviously, given a decoupling cycle, one gets
$$
\frac{1}{N} \sum_{k=1}^N v_k x v_k^* = \frac{1}{|V|} \sum_{v\in V} v x v^*.
$$

Dynamical decoupling on $\H_s\t\H_e$ is now implemented by applying the decoupling operations $v\t\unit_{\H_e}$ instantaneously in time steps $\tau>0$. To shorten notation, we shall simply write $v$ instead of $v\t\unit_{\H_e}$ when confusion is unlikely. In~\cite{ABH15,ABH17} we discuss a random implementation of these decoupling operations while here we restrict ourselves to a deterministic implementation since our focus is rather on the unboundedness of $H$. To be precise, consider a decoupling cycle  of unitaries $(v_1,v_2,\ldots,v_N)$ and apply them to the system periodically, so the total time evolution unitary after one decoupling cycle will be given by
\begin{equation}\label{eq:decevo}
 v_N \rme ^{\rmi \tau H}v_N^* v_{N-1} \rme ^{\rmi \tau H}v_{N-1}^* \cdots v_2\rme ^{\rmi \tau H}v_2^* v_1 \rme ^{\rmi \tau H}v_1^*.
\end{equation}
We can now split a given time interval $[0,t]$ into $nN$ steps and apply the decoupling cycle of length $N$ there $n$ times. Thus the following definition makes sense:

\begin{definition}\label{def:ddworks}

For given Hamiltonian $H$ and decoupling set $V$, we say that \emph{dynamical decoupling works specifically} if there is a decoupling cycle 
$(v_1, v_2,\ldots,v_N)$ and there is a self-adjoint operator $B$ on $\H_e$ such that
\[
\slim_{n\ra\infty} \Big( \rme ^{\rmi \frac{t}{n N} v_1Hv_1^*} \cdots \rme ^{\rmi \frac{t}{n N} v_NHv_N^*}\Big) ^n = \rme^{\rmi t (\unit_{\H_s}\t B)},
\]
uniformly for $t$ in compact intervals of $\R$.

We say that \emph{dynamical decoupling works uniformly} if there is a self-adjoint operator $B$ on $\H_e$ such that, for every decoupling cycle $(v_1, v_2,\ldots,v_N)$,
\begin{equation}\label{eq:decworks}
\slim_{n\ra\infty} \Big( \rme ^{\rmi \frac{t}{n N} v_1Hv_1^*} \cdots \rme ^{\rmi \frac{t}{n N} v_NHv_N^*}\Big) ^n = \rme^{\rmi t (\unit_{\H_s}\t B)},
\end{equation}
uniformly for $t$ in compact intervals of $\R$.
\end{definition}

The physical interpretation is that, in the limit where time steps go to $0$, only the environment evolves. From the physical point of view the strong topology is satisfactory. Indeed, one gets norm convergence (that is uniform rate) on $\H_s$ for  a fixed environment state $\rho_e$, for example a thermal state.

It is unclear whether ``specifically" in Definition \ref{def:ddworks} is really weaker than ``uniformly" or whether the existence of one decoupling cycle which works would in fact imply that all decoupling cycles work. Intuitively, one might  expect that the order is irrelevant  as a consequence of the homogenisation effect of the limit.

To conclude the prerequisites, we will frequently use the following convention: if $a_k$, with $k=1,\ldots,N$, are in $B(\H)$ then we write
\[
\prod_{k=1}^N a_k := a_N \cdots a_1
\]
for the product in $B(\H)$ with this specific order.

\section{A sufficient condition for dynamical decoupling}\label{sec:suff}

\begin{theorem}\label{th:suff}
Let $V$ be a decoupling set for $\H_s$, and $H: \D(H) \to \H$ be self-adjoint. 
If the sum $\sum_{v\in V} (v\t \unit_{\H_e}) H (v\t\unit_{\H_e})^*$ is essentially self-adjoint on  the intersection of the domains, $\D= \bigcap_{v\in V} v \D(H)$, then dynamical decoupling works uniformly for $H$.
\end{theorem}

\begin{proof}
The theorem follows from a straight-forward generalisation of the Trotter product formula~\cite[Cor.11.1.6]{JL} to $N$ factors. More precisely, given a decoupling cycle $(v_1,\ldots,v_N)$ in $V$, define the function
\[
F(t) := \prod_{k =1}^{N} \rme ^{\rmi\frac{t}{N} v_k Hv_k^*}, \quad t\in\R_+.
\]
Then $F:\R\ra B(\H)$ is a strongly continuous function with $F(0)=\unit_{\H}$. Moreover, we get
\begin{align*}
\frac{F(t)\xi-F(0)\xi}{t} 
=& \frac{\prod_{k} \rme ^{\rmi\frac{t}{N} v_k Hv_k^*} \xi - \xi}{t}
\ra \frac{\rmi}{|V|}\sum_{v\in V} vHv^*\xi, \quad t\ra 0\\
\end{align*}
for all $\xi\in\D$.

Now, we claim that the closure
\[
\frac{1}{|V|}\overline{\sum_{v\in V} vHv^*} = \unit_{\H_s}\t B,
\]
with some  self-adjoint $B$ on $\H_e$. 
Indeed, since the left-hand side is commuting with all $v\in V$, the group $(U_t)_{t\in\R}$ generated by it satisfies the relation
$$
U_t = \frac{1}{|V|} \sum_{v\in V} v U_t v^* =  \unit_{\H_s} \otimes \frac{1}{d}  \tr_{\H_s}(U_t), \qquad \text{for all } t\in\R, 
$$
by Definition~\ref{def:dd} of decoupling set, and thus must be of the form 
\[
U_t =  \unit_{\H_s}\t \rme^{\rmi B t},
\]
by Stone's theorem.

We apply Chernoff's product formula~\cite[Thm.~11.1.2]{JL} to this and obtain that
\[
F(t/n)^n = \Big(\prod_{k=1}^N \rme ^{\rmi \frac{t}{N n} v_kHv_k^*}\Big)^n \ra \rme^{\rmi (\unit_{\H_s}\t B) t},\quad n\ra \infty,
\]
in the strong operator topology and uniformly for $t$ in compact intervals in $\R$. This verifies condition~\eqref{eq:decworks}.
\end{proof}

\begin{example}[Qubit]\label{ex:qubit}

The following construction is a building block that will allow us to create several examples at increasing complexity and transfer results about the Trotter formula to the context of dynamical decoupling.
The idea is to study the space 
$$\H=\C^2\t \H_e\simeq \H_e\oplus \H_e,$$ 
describing a qubit system coupled to an environment $\H_e$. 
Suppose our Hamiltonian, expressed in the decomposition $\H_e\oplus \H_e$ of $\H$, is of the form $A\oplus B$, i.e.\
\begin{equation}
\label{eq:AB}
H=\begin{pmatrix}
A & 0\\
0 & B
\end{pmatrix}
\end{equation}
on $\D(A)\oplus\D(B)$, with both $A,B$ self-adjoint. The standard decoupling set for $\C^2$ consists of the Pauli group: (multiples by $1,\rmi,-1,-\rmi$) of the four Pauli matrices $\unit,X,Y,Z$. Now if we take the Pauli matrix
\[
X=\begin{pmatrix}
0 & 1\\
1 & 0
\end{pmatrix}
\]
then
\[
H=\begin{pmatrix}
A & 0\\
0 & B
\end{pmatrix}, \quad
X H X^* = \begin{pmatrix}
B & 0\\
0 & A
\end{pmatrix}.
\]
The adjoint action of other Pauli matrices to $H$ produces one of these two matrices, so we can reduce our situation down to a group with two elements $V=\{\unit,X\}$. As decoupling cycles in the examples here we consider simply $(\unit,X)$ although an analogous reasoning holds for any other cycle in $V$. Thus though we prove everything only for this specific cycle, one can actually show that decoupling works uniformly in all of the following examples.
\hfill\boxy
\end{example}

\begin{example}[Shallow-pocket model]\label{ex:sp}
See~\cite[Sec.3]{AHFB15}.
In the setting of the preceding Ex.~\ref{ex:qubit} with one qubit, we consider $\H_e = L^2(\R)$ and  $A=-B=q$,
the position operator, $q\xi(x)= x\xi(x)$, with $\D(q)= \{\xi \in L^2(\R): q \xi \in L^2(\R) \}$, in~(\ref{eq:AB}):
\begin{equation}
\label{eq:SP}
H=\begin{pmatrix}
q & 0\\
0 & -q
\end{pmatrix}.
\end{equation}
Thm.~\ref{th:suff} applies, and the model can be dynamically decoupled. In fact we get that $XHX=-H$, so the Trotter limit is trivial and decoupling works uniformly and perfectly at all time scales.

We can study the reduced dynamics as well. Let us assume the environment initial state is 
\begin{equation}
\xi_C(x) = \left(\frac{2}{\pi} \frac{1}{x^2 +4}  \right)^{1/2}, \quad x\in\R.
\label{eq:Cauchy}
\end{equation}
The  spectrum of $H$ is the full line $\R$. The state $\xi_C$ does not belong to the domain $D(q)$, and then any initial factorized state $\psi \otimes \xi_C$ does not belong to the domain of $H$.

The dynamical map $\Lambda_t$ in~(\ref{eq:dynamical_map}) has the semigroup property and its generator is the dephasing GKLS operator
\[
\L \rho = -[Z,[Z, \rho]],
\]
for all system density matrices $\rho$, where 
\[
Z=\begin{pmatrix}
1 & 0\\
0 & -1 
\end{pmatrix} ,
\]
giving rise to exponential decay of the coherences. See Fig.~\ref{fig:Exponential-and-non-exponential}. \hfill\boxy
\end{example}

\begin{example}[$q\oplus p$  example]\label{ex:x-p}
Again in the setting of Ex.~\ref{ex:qubit}, we choose $\H_e = L^2(\R)$ and  $A= q$ and $B=p$,
the position and the momentum operator, respectively,  in~(\ref{eq:AB}):
\begin{equation}
\label{eq:qp}
H=\begin{pmatrix}
q & 0\\
0 & p
\end{pmatrix},
\end{equation}
where $q= M_x$ and $p=-\rmi \rmd/\rmd x$ are self-adjoint on their natural domains $\D(q)= \{\xi \in L^2(\R): q \xi \in L^2(\R) \}$, $\D(p)=H^1(\R)$, the first Sobolev space. The sum of the two is essentially self-adjoint, e.g. on the Schwartz space $\mathcal{S}(\R)$. According to Ex.~\ref{ex:qubit}, it is sufficient to consider the group $V=\{\unit,X\}$, and
\[
H+XHX^* = \unit_{\C^2} \otimes (q+p)
\]
which is essentially self-adjoint. According to Thm.~\ref{th:suff} dynamical decoupling works uniformly.

We can study the reduced dynamics as follows. As environment initial state let us consider 
\[
\xi = \rme^{-\rmi \frac{\pi}{8} (q^2+p^2)} \xi_C,
\]
with $\xi_C$ as in~\eqref{eq:Cauchy}.
Then the dynamical map $\Lambda_t$ in~(\ref{eq:dynamical_map})  is generated by the dephasing GKLS operator plus a time dependent Hamiltonian
\[
\L_t \rho = -[Z,[Z, \rho]] - \rmi \frac{t}{2} [Z, \rho],
\]
giving rise to exponential decay of the coherences. 

In order to determine the unitary evolution after $n$ decoupling cycles and the decoupling error explicitly, we need some prerequisites. Consider the 3-dimensional real Lie algebra $\lg= \spa_{\R} \{E, P, Q\}$ with commutation relations
\[
[E,P]=0,\quad [E,Q] = 0,\quad [P,Q] = E,
\]
and its representation $\pi$ by unbounded skew-symmetric operators defined by linear continuation of
\[
E \mapsto \rmi \unit, \quad P\mapsto \rmi p, \quad Q\mapsto \rmi q
\]
where all operators here act on the Schwartz space $\Sch(\R)\subset L^2(\R)$ as common invariant domain. Notice that $L^2(\R)$ with Hamiltonian $p^2+q^2$ is a representation of the harmonic oscillator. The subspace $\A$ of finite energy vectors (i.e., the span of eigenvectors of $p^2+q^2$) forms joint analytic vectors for $p$ and $q$ because all monomials have an ``energy bound" in terms of $p^2+q^2$, i.e., $\|X_1\cdots X_m \xi_n\| \le \|(\unit+p^2+q^2)^{m/2}\xi_n\|=n^{m/2}\|\xi_n\|$, for $n$ the ``energy" of the eigenvector $\xi_n\in\A$ and all monomials with $X_i\in\{\unit,p,q\}$. Moreover, $\A$ is dense in $L^2(\R)$, and the conditions in Nelson's criterion \cite[Lem.9.1]{Nel} are fulfilled. Therefore the representation $\pi$ exponentiates to the Lie group $G$ of $\lg$ such that
\[
\rme^{\pi(X)} = \pi(\exp(X)),
\]
where $\exp$ denotes the exponential map of $G$. This means that the Baker-Campbell-Hausdorff formula holds in the representation $\pi$ as well, namely
\begin{align*}
 \rme ^{\rmi q} \rme ^{\rmi p} =& \pi(\exp(Q)) \pi(\exp(P))= \pi(\exp(Q)\exp(P)) = \pi(\exp(Q+P + \frac12 [Q,P])) \\
 =& \pi(\exp(Q+P-\frac12 E)) =  \rme^ {\rmi (q+p)- \frac{\rmi}{2} };
\end{align*}
all higher order commutators in $\lg$ vanish.

We apply this now to dynamical decoupling. The time evolution after time $t$ with $n$ decoupling cycles of length $N=|V|=2$ then reads
\begin{equation}
U_n (t) = \Big( \rme ^{\rmi \frac{t}{2 n} H}\rme ^{\rmi \frac{t}{2 n} XHX}\Big)^n =\exp \left( - \rmi  \frac{t^2}{8n} Z  \right) \otimes \exp\left(\rmi \frac{t}{2}  (q+p)\right) . 
\end{equation} 
Interestingly the decoupling error---the deviation from Eq. (\ref{eq:decworks})---is a unitary on the system only, and the convergence is uniform for $t$ in compact intervals in $\R$:
\[
\|U_n(t) - \unit_{\H_s}\t \rme^{\rmi B t}\| = \|  \rme ^{ -\rmi  \frac{t^2}{2n} Z} - \unit_{\H_s}\| \to 0,
\]
as $n \to +\infty$, with $B=\frac12 (q+p)$.
\hfill\boxy
\end{example}

\begin{example}[$q\oplus p^2$ example] \label{ex:x-p2} 
Again in the setting of Ex.~\ref{ex:qubit}, we now choose $\H_e = L^2(\R)$ and  $A= q$ and $B=p^2$, in~(\ref{eq:AB}):
\begin{equation}
\label{eq:qp2}
H=\begin{pmatrix}
q & 0\\
0 & p^2
\end{pmatrix}.
\end{equation}
The operators $q=M_x$ and $p^2=-\rmd^2/\rmd x ^2$ are self-adjoint on their natural domains $\D(q)= \{\xi \in L^2(\R): q \xi \in L^2(\R) \}$, $\D(p^2)=H^2(\R)$, the second Sobolev space. The sum of the two is essentially self-adjoint on the Schwartz space $\mathcal{S}(\R)$. According to Ex.~\ref{ex:qubit}, it is sufficient to consider $V=\{\unit,X\}$, and
\[
H+XHX^* = \unit_{\C^2} \otimes (q+p^2)
\]
which is essentially self-adjoint. According to Thm.~\ref{th:suff} dynamical decoupling works uniformly.

In order to study the decoupling error, let us consider a Cauchy distribution in momentum space 
\[
\xi(p) =\left(\frac{\gamma}{2\pi}\frac{1}{p^{2}+\frac{\gamma^2}{4}} \right)^{1/2}, \quad p\in\R,
\]
as environment initial state. Then for the qubit state $\rho_{s}(t)$ at time $t$, the decoupling error becomes
\[
\epsilon(t)=\Vert \rho_{s}(t)-\rho_{s}(0)\Vert_{2}^{2},	
\]
as a function of the decoupling steps $n$ with $\Vert \cdot \Vert_{2}$ being the Hilbert Schmidt norm and we assume that the qubit is initially prepared in $\rho_{s}(0)=|+\rangle\langle+|$, where $|+\rangle= (\frac{1}{\sqrt{2}},\frac{1}{\sqrt{2}})$ yielding $\epsilon(t)=2(1-\langle+|\rho_{s}(t)|+\rangle)$. 

We proceed in analogy with the previous example, verifying Nelson's criterion for the Lie algebra $\lg=\spa_\R \{E,P,Q,R\}$ with commutation relations
\[
[E,P]=[E,R]=[E,Q]=0,\quad [P,R]=0,\quad [P,Q] = E, \quad [R,Q]= 2P.
\]
We represent $\lg$ by unbounded skew-symmetric operators
\[
E \mapsto \rmi \unit, \quad P\mapsto \rmi p,  \quad Q\mapsto \rmi q, \quad R\mapsto \rmi p^2
\]
on $\Sch(\R)$. All monomials in $p,p^2,q$ can again be energy-bounded in terms of $p^2+q^2$ on $\A$. Thus the Lie algebra representation exponentiates to a Lie group representation again and the Baker-Campbell-Hausdorff and the Zassenhaus formula hold. Since nested commutator expressions vanish after depth $3$, the Baker-Campbell-Hausdorff and the Zassenhaus formula show that the unitary evolution after $n$ decoupling cycles $U_{n}(t)$ takes the form
\[
U_{n}(t)=\rme^{-\rmi \frac{t}{2}(\unit\otimes (q+p^2)+t^{2}/(24n^{2}))}  \rme^{-\rmi \frac{t^{3}}{16n}(Z\otimes \unit)}  \rme^{\rmi \frac{t^{2}}{4n}(Z\otimes p)}.
\]
 Since this evolution, for finite $n$, leads to dephasing in $Z$ direction of the qubit, we remark here that the choice for $\rho_{s}(0)$ as above describes the worst-case scenario, i.e. the supremum of $\epsilon$ over all initial states of the qubit. After tracing out the environmental degrees of freedom we obtain
\begin{equation}\label{eq:explicitformdece}
\epsilon(t)=1-\cos(t^{3}/(16n))\rme^{\frac{-t^{2}}{4n}\gamma},	
\end{equation}
which vanishes for $n\to\infty$. We thus have found the explicit form of the decoupling error for the $q\oplus p^2$ model, which is plotted as a function of $t$ and $n$ in Fig.~\ref{fig:decerror} for a fixed $\gamma=1$. 
\begin{figure}[!h]
\includegraphics[width=0.45\columnwidth]{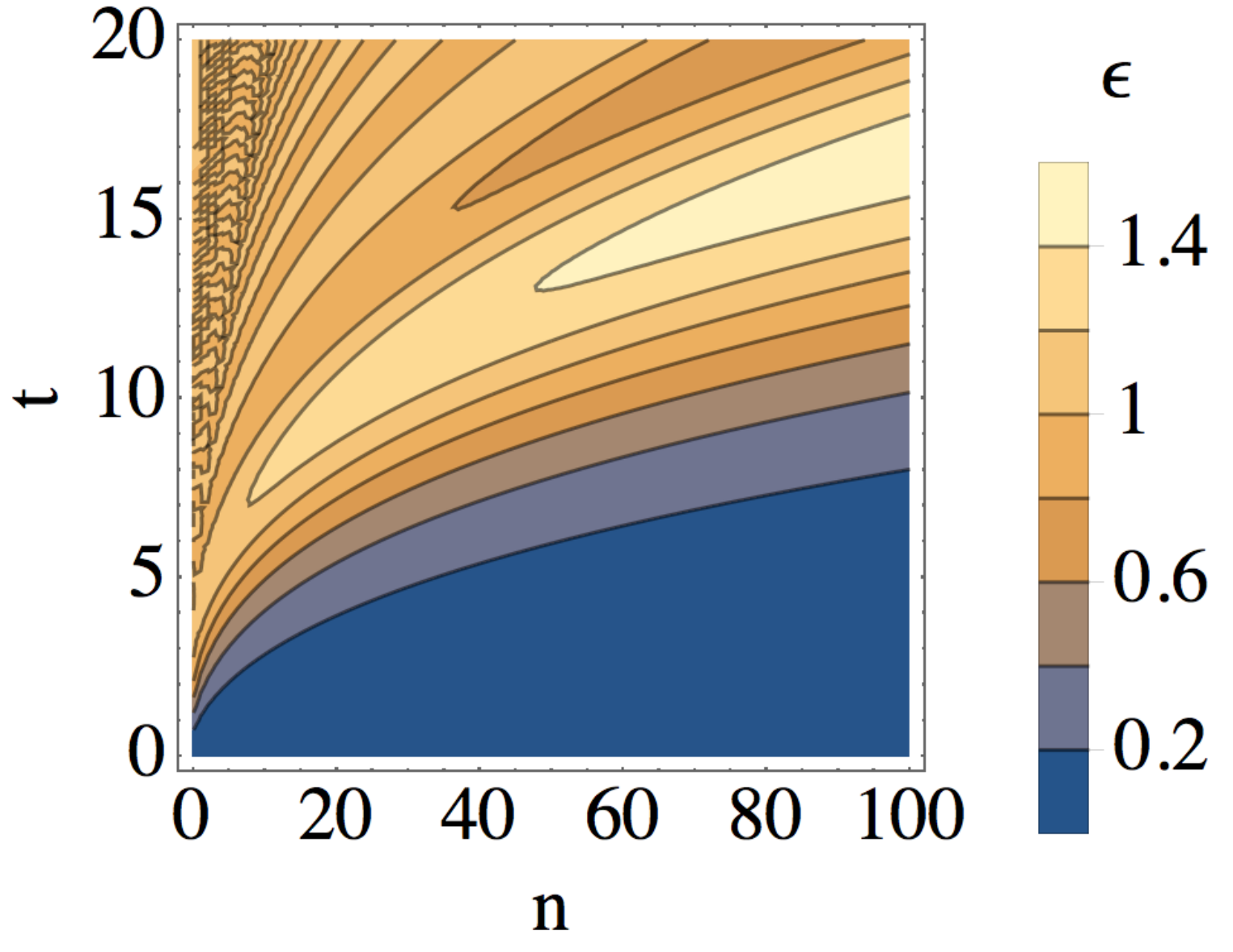}
\caption{\label{fig:decerror} Decoupling error~\eqref{eq:explicitformdece} for the $q\oplus p^2$ model given by~\eqref{eq:qp2} as a function of the total evolution time $t$ and the decoupling steps $n$ for a fixed $\gamma=1$.}	
\end{figure}
 The form of the decoupling error~\eqref{eq:explicitformdece} also shows the dependency of the initial state of the environment on the efficiency of dynamical decoupling. 
For fixed $t>0$ and $n<\infty$ we can always find an environment initial state, i.e. some $\gamma>0$,  such that decoupling becomes arbitrarily bad.   
\end{example}

\begin{example}[$q^2\oplus p^2$ example]\label{ex:x2-p2}
We choose $\H_e = L^2(\R)$ and  $A= q^2$ and $B=p^2$, in~(\ref{eq:AB}):
\begin{equation}
\label{eq:q2p2}
H=\begin{pmatrix}
q^2 & 0\\
0 & p^2
\end{pmatrix}
\end{equation}
$q^2$ and $p^2$ are self-adjoint on their natural domains $\D(q^2)= \{\xi \in L^2(\R), q^2 \xi \in L^2(\R) \}$, $\D(p^2)=H^2(\R)$, the second Sobolev space. The sum of the two is essentially self-adjoint on the Schwartz space $\mathcal{S}(\R)$.
Thm.~\ref{th:suff} applies, and the model can be dynamically decoupled. 

The unitary evolution after $n$ decoupling steps can be obtained using a symplectic representation \cite{symplectic}, and reads for $n>t$ 

\begin{equation}
U_n (t)  =  \exp\left(\rmi t f\left(\frac{t}{n}\right) \Big(\overline{\unit\otimes  (q^2+p^2) + \frac{t}{2n} X\otimes (q p+p q)}\Big)\right) ,
\end{equation} 
where
\[
f(t) = \frac{2}{t \sqrt{4-t^2}} \arctan\left(\frac{t\sqrt{4-t^2}}{2-t^2}\right).
\]
Notice that $f(t) = 1 +t^2/6 + O(t^4)$ for $t\to0$.

In this case the decoupling error (on the system) is not unitary, and the convergence is non-uniform. Interestingly, the original Hamiltonian has an absolutely continuous spectrum on the positive real line, while the limit is the Harmonic oscillator, which has a purely point spectrum.

\hfill\boxy
\end{example}
\begin{example}[Spin-boson model]\label{ex:CL}
Consider again $\H_s=\C^2$ and $\H_e = L^2(\R)$ as in the preceding examples, but now
$$
H= \overline{\omega _{c}  \unit\otimes  a^\dagger a + \frac{\omega_{a}}{2} Z \otimes \unit 
+ \frac{\Omega}{2} \left(\sigma_{+}  \otimes a + \sigma_{-}  \otimes a^\dagger \right)}.
$$
Here $a=\frac{1}{\sqrt{2}}(q+\rmi p)$ and $a^\dagger=\frac{1}{\sqrt{2}}(q-\rmi p)$ (the formal adjoint of $a$) are the harmonic oscillator ladder operators on the common invariant domain $\Sch(\R)$; $\omega_c$, $\omega_a$ and $\Omega$ are real constants,
$$ 
\sigma_+ = \begin{pmatrix}
0 & 1\\
0 & 0 
\end{pmatrix} , \qquad
\sigma_- = \begin{pmatrix}
0 & 0\\
1 & 0 
\end{pmatrix} ,
$$
and $\overline{X}$ denotes the closure of a closable operator $X$ on $\C^2\otimes L^2(\R)$, so that $H$ is self-adjoint.
This way the interaction part of the Hamiltonian is not block-diagonal, in contrast to Ex.~\ref{ex:qubit}.
This model can be decoupled using the full Pauli group $V$ because the sum $\sum_{v\in V} v H v^*$ is essentially self-adjoint on  the Schwartz space $\mathcal{S}(\R)$ which is contained in the domain intersection. 
This example can be easily generalized to a countable number of bosonic modes.
\hfill\boxy
\end{example}

One could argue whether the conditions of Thm.~\ref{th:suff} are also necessary. This is not the case, as shown in the following slightly artificial example by Chernoff~\cite{Che}. Later we will give a necessary condition for semibounded Hamiltonians.

\begin{example}[Non-overlapping domains]
\label{rm:cher}
We use again the setting Ex.~\ref{ex:qubit}. In~(\ref{eq:AB}), let $\H_e = L^2(\R)$ and consider $A=p,$ the momentum operator, 
which is self-adjoint  on $\D(p)= H^1(\R)$, the first Sobolev space. Let $B$ be the multiplication operator 
$$B \psi (x) = f(x) \psi(x), \qquad \D(B) = \{ \psi\in L^2(\R), f \psi \in L^2(\R)\},$$
where $f\in L^1_{\mathrm{loc}}(\R) \setminus L^2_{\mathrm{loc}}(\R)$ is locally integrable, but is not locally square-integrable.
For example,  we can take
\[
f(x) =\sum_{n=1}^\infty \frac{1}{n!} |x- r_n|^{-1/2},
\]
where $(r_n)_{n\in\N}$ is some enumeration of the rationals. One can prove that~\cite[Prop.~5.1]{Che}
$$
\slim_{n\to+\infty} \left(\rme^{\rmi \frac{t}{n} p}  \rme^{\rmi \frac{t}{n} B} \right)^n = \rme^{-\rmi C}  \rme^{\rmi t p}  \rme^{\rmi C}, 
$$
where $C$ is the multiplication operator 
$$C \psi (x) = F(x) \psi(x), \qquad \D(C) = \{ \psi\in L^2(\R), F \psi \in L^2(\R)\},$$
by an antiderivative of $f$,
$$F(x) = \int_0^x f(y) \rmd y.$$

The unitary evolution after $n$ decoupling cycles $(\unit,X)$ converges 
\begin{equation}
\slim_{n\to+\infty} U_n (t) = \slim_{n\to+\infty} \Big( \rme ^{\rmi \frac{t}{2n} H}\rme ^{\rmi \frac{t}{2n} XHX}\Big)^n =\unit_{\H_s} \otimes \exp \left(\rmi t\rme^{-\rmi C}  p  \rme^{\rmi C} \right) ,
\end{equation} 
and therefore decoupling works specifically, and a similar reasoning applies to any decoupling cycle, so dynamical decoupling works uniformly.

Notice that if $\psi \in C(\R)$ is continuous, and thus locally bounded, then $f \psi \notin L^2_{\mathrm{loc}}(\R)$. Therefore, $\D(B)$ does not contain any nonzero continuous function, and thus 
$$
\D(A)\cap \D(B) = \{0\},
$$
since $H^1(\R) \subset C(\R)$. 
Thus the Trotter formula for unitary groups can converge  when the operator sum is not essentially self-adjoint, and even in the extreme case of a trivial domain intersection. Dynamical decoupling works even though $\D= \bigcap_{v\in V} v \D(H)$=\{0\}.

\qed
\end{example}

\section{A necessary condition for dynamical decoupling of non-negative Hamiltonians}\label{dec:nec}

In Thm.~\ref{th:suff} we have established a sufficient condition for dynamical decoupling to work uniformly. However, Ex.~\ref{rm:cher} showed that it is not at all necessary, so let us now turn to our promised necessary condition, under the additional assumption of a non-negative Hamiltonian:

\begin{theorem}\label{th:nec}
Let $V$ be a decoupling set for $\H_s$, and suppose that $H$ is non-negative. If dynamical decoupling works uniformly then for all $v,w\in V$, the form domain intersections
\[
v \D(H^{1/2}) \cap w \D(H^{1/2}) \subset \H
\]
must be dense.  
\end{theorem}

Before starting the proof, let us quickly recall something about form domains. Every densely defined operator $A$ on $\H$ gives rise to a bilinear form $\D\times\D\ra\C$ with some form domain $\D\subset\H$, in general not unique. In the case $A$ is non-negative, this form domain is defined as $\D(A^{1/2})$. Notice that $\D(A^{1/2})\supset \D(A)$, so the form domain of (the bilinear form of) a non-negative operator is always dense. Given two non-negative operators, which might have trivial domain intersection and therefore no sum but whose form domains intersect densely, it is possible to define a sum of the two forms; this new bilinear form corresponds to a new self-adjoint operator which is generally called the \emph{form sum} of the two initial operators. For a proper introduction to form domains we refer the reader to \cite[Sec.8.6]{Reed} or \cite[Sec.10.3]{JL}.

\begin{proof}
Suppose that dynamical decoupling works but not all of the form domain intersections are dense. Choose $v_1,v_2\in V$ such that for these two elements,
\[
\H_0 := \overline{v_1 \D(H^{1/2}) \cap v_2 \D(H^{1/2})} \not= \H .
\]
Choose $\xi \in \H_0^\perp$ with $\|\xi\|=1$. Extend the two elements $v_1,v_2$ to a palindromic cycle of length $N=2|V|$ in $V$, say $(v_1, v_2, \ldots, v_{|V|}, v_{|V|}, \ldots v_2,v_1)$. We can then define the following continuous functions
\[
f_n: \overline{\C_+} \ra \H, \quad f_n(z)= \Big(\rme ^{-\frac{z}{n N} v_1Hv_1^*} \cdots \rme ^{-\frac{z}{n N} v_{|V|} H v_{|V|}^*} \rme ^{-\frac{z}{nN} v_{|V|}Hv_{|V|}^*} \cdots \rme ^{-\frac{z}{nN} v_1Hv_1^*}\Big)^n \xi
\]
which are analytic on $\C_+$, for every $n\in\N$; here $\C_+=\{z\in\C: \Re(z) > 0\}$ denotes the open complex right half-plane. Since we assumed dynamical decoupling to work uniformly, we know from~\eqref{eq:decworks} that $f_n$ converges on the boundary and there is a selfadjoint $B$ on $\H_e$ such that
\[
f_n(-\rmi t) \ra  \rme^{\rmi t (\unit_{\H_s}\t B)} \xi,
\]
as $n\ra\infty$, uniformly for $t$ in compact intervals in $\R$. Moreover,~\cite[Thm.~7.2]{Che} shows that, since $vHv^*$ is non-negative for every $v\in V$, $B$ is non-negative as well and 
\begin{equation}\label{eq:formproof2}
f_n(z)\ra f(z):=\rme^{-z (\unit_{\H_s}\t B)}\xi, \quad n\ra\infty,
\end{equation}
uniformly for $z$ in compact subsets of $\C_+$; and $f$ is continuous on $\overline{\C_+}$ and analytic on $\C_+$. It is obvious that
\begin{equation}\label{eq:formproof0}
f(0)=\xi\not= 0. 
\end{equation}

We now claim that
\begin{equation}\label{eq:formproof1}
f_n(t) 
 \ra 0
\end{equation}
as $n\ra\infty$, uniformly for $t$ in compact intervals in $(0,\infty)$.
To this end we make use of the proof in~\cite{Kato}. Following the notation there, let us write
\[
F_t'= \rme ^{- t v_1Hv_1^*} \rme ^{- t v_{2}Hv_{2}^*} \rme ^{- t v_{2}Hv_{2}^*} \rme ^{- t v_1Hv_1^*}, \quad t\in[0,\infty).
\]
This is precisely the quantity defined in~\cite[(3.6)]{Kato}. Moreover, let us define
\[
G_t= \rme ^{-\frac{t}{N} v_1Hv_1^*} \cdots \rme ^{-\frac{t}{N} v_{|V|}Hv_{|V|}^*} \rme ^{-\frac{t}{N} v_{|V|}Hv_{|V|}^*} \cdots \rme ^{-\frac{t}{N} v_1Hv_1^*}, \quad t\in[0,\infty).
\]
Then it follows that $0\le G_t\le F_{t/N}'\le \unit$, for all $t\in [0,\infty)$, and
\[
0\le G_t  = \unit - (\unit-G_t) \le (\unit + (\unit-G_t))^{-1}.
\]
We are interested in the limit of $G_{t/n}^n\xi$. We have
\[
0\le G_{t/n}^{2n} \le (\unit + (\unit-G_t))^{-2n} \le (\unit + 2n (\unit-G_t))^{-1}.
\]
Using the fact that $x\mapsto -\frac{1}{x}$ is operator-monotone on $(0,\infty)$~\cite{HOT} and the fact that $0\le G_t\le F_{t/N}'\le \unit$, so $\unit \le \unit + 2n (\unit-F_{t/N}') \le \unit + 2n (\unit-G_t)$, we get
\[
(\unit + 2n (\unit-G_t))^{-1} \le  (\unit + 2n (\unit-F_{t/N}'))^{-1}.
\]
Then it follows from~\cite[(3.11)]{Kato}
\[
\|G_{t/n}^n\xi\|^2 = \langle \xi, G_{t/n}^{2n}\xi \rangle \le \langle \xi, (\unit + 2n (\unit-F_{t/N}'))^{-1} \xi \rangle \ra 0, \quad n\ra\infty,
\]
uniformly for $t$ in compact intervals in $(0,\infty)$. This proves our claim in~\eqref{eq:formproof1}, namely $f_n(t)\ra 0$, uniformly for $t$ in compact intervals of $(0,\infty)$.

On the other hand,~\eqref{eq:formproof2} shows that $f_n(t)\ra f(t)$ as $n\ra\infty$, which means that $f(t)=0$, for $t\in(0,\infty)$. By the identity theorem for analytic functions, we get that $f(z)=0$, for all $z\in\C_+$, and since $f$ is continuous on $\overline{\C_+}$, we must have $f(0)=0$ as well. This is in contradiction with~\eqref{eq:formproof0}. Thus dynamical decoupling cannot work uniformly if the form domain intersections are not dense.
\end{proof}

The preceding theorem provides a necessary condition for dynamical decoupling to work uniformly, namely that the form domain intersections $v \D(H^{1/2}) \cap w \D(H^{1/2})$ are dense in $\H$, for every two $v,w\in V$. We believe that it should be possible to strengthen this as follows, though in order to prove this we would require a generalisation of~\cite{Kato} to Trotter products of arbitrarily many semigroups rather than only two, which is currently an open problem.

\begin{conjecture}\label{conj:multiprod}
Let $V$ be a decoupling set for $\H_s$, and suppose that $H$ is non-negative. If dynamical decoupling works uniformly then the total form domain intersection
\[
\bigcap_{v\in V} v \D(H^{1/2})  \subset \H
\]
must be dense.  
\end{conjecture}

In the case where $V$ consists of two elements, the conjecture reduces to Thm.~\ref{th:nec}, and we can realize the relevance of the condition in the following example.

\begin{example}[Non-overlapping form domains]\label{ex:Lapidus}
Assume that in~(\ref{eq:AB}) both $A,B\ge 0$ but vanishing form domain intersection, $\D(A^{1/2})\cap \D(B^{1/2})=\{0\}$. 
Then applying the decoupling operations on $\C^2$ as in Ex.~\ref{ex:qubit} leads to
\[
\D(H^{1/2})\cap\D(vH^{1/2}v^*) = \big(\D(A^{1/2})\oplus \D(B^{1/2})\big) \cap \big(\D(A^{1/2})\oplus \D(B^{1/2})\big) = \{0\}.
\]
According to the criterion in Thm.~\ref{th:nec} this system cannot be decoupled from the  environment.

Now in order to find such operators, let us modify Ex.~\ref{rm:cher}, see also \cite[Ex.5.6]{Che} or \cite[Ex.10.3.21]{JL}. Namely, consider $\H_e= L^2(\R)$, and $A=p^2$ the negative second derivative operator $-\frac{\rmd^2}{\rmd x^2}$ on $\R$ and $B$ the multiplication with a certain positive measurable function $f$ yet to be determined. The domain of $A$ is the second Sobolev space,
$\D(A) = H^2(\R)$, 
and the form domain is the first Sobolev space, 
$\D(A^{1/2}) = H^1(\R)$.  
Instead for $B$ we find
\[
\D(B) = \{\xi\in L^2(\R): f\xi \in L^2(\R)\}
\] 
and
\[
\D(B^{1/2}) = \{\xi\in L^2(\R): \sqrt{f}\xi \in L^2(\R)\}.
\]
Now we take $f$ in such a way that it is nowhere locally integrable. E.g., we can take
\[
f(x) = \Big( \sum_{n=1}^\infty \frac{1}{n!} |x- r_n|^{-1/2} \Big)^2,
\]
where $(r_n)_{n\in\N}$ is a complete enumeration of the set of rational numbers.

With this choice, one can prove that $A$ and $B$ are densely defined self-adjoint operators on $L^2(\R)$ but with trivial form domain intersection:
\[
\D(A^{1/2})\cap D(B^{1/2})=\{0\},
\]
which concludes our example. 
\hfill\boxy
\end{example}

\begin{remark}[Some variations]\label{rem:var}
In order to allow for a wider selection of models which can be dynamically decoupled, we might try to relax the condition of dynamical decoupling in~\eqref{eq:decworks} a bit. One way forward would be to say that as we fix $t$ and let $n\ra\infty$, we no longer require~\eqref{eq:decworks} for the whole sequence but instead for a subsequence only. In other words, we could say that dynamical decoupling works if, for any $t\in [0,\infty)$,
there is a subsequence $(n_k)_{k\in\N}$
such that
\begin{equation}
\label{eq:weakDD}
\slim_{k\ra\infty} \Big( \prod_{j=1}^N \rme ^{\rmi\frac{t}{n_k N} v_j Hv_j^*}\Big) ^{n_k} = \rme^{\rmi t (\unit_{\H_s}\t B)}.
\end{equation}
Interestingly enough, it follows from \cite{KM} using the argument in \cite[Prop.11.7.4]{JL} that for any system with decoupling set $V$ and $H\geq0$, the condition of dense total form domain intersection in Conj.~\ref{conj:multiprod} is sufficient in order for~(\ref{eq:weakDD}) to hold for almost all $t$ (though \emph{not} for all $t$ and \emph{not} uniformly in compact intervals).
\end{remark}

We end this section with a difficult open problem:

\begin{problem}\label{prob2}
Let $V$ be a decoupling set for $\H_s$, and suppose that $H$ is non-negative. Is it true that dynamical decoupling works (uniformly) if and only if the total form domain intersection
\[
\bigcap_{v\in V} v \D(H^{1/2})  \subset \H
\]
is dense? 
\end{problem}

Notice that an affirmative answer would prove Conj.~\ref{conj:multiprod}.

\section{Generalized Friedrichs-Lee model}
\label{ex:Lee}

The aim here is to provide a physically realistic model where dynamical decoupling does not work.
The environment is described by the standard Friedrichs-Lee model, which is defined in detail in~\cite[Sec.4.2.2-3]{Wal} and will be quickly recalled. In this section, the Hamiltonian will no longer be non-negative, so Thm.~\ref{th:nec} becomes irrelevant here.

The Hilbert space is  $\H_e=\C \oplus L^2(\R)$. It will be convenient to use a matrix notation and write the vectors $\psi = x \oplus \xi \in \H_e$ in the form
\begin{equation}
\psi(\omega) =
\begin{pmatrix}
x \\ \xi(\omega)
 \end{pmatrix},
\label{eq:vector}
\end{equation}
where $\omega\in\R$, $x\in \C$, and $\xi \in L^2(\R)$.
Then, given a function $g\in L^2(\R)$, the Friedrichs-Lee Hamiltonian has a block matrix form defined by~\cite{Friedrichs}
\begin{equation}
\label{eq:matrixaction}
(H_g \psi) (\omega)  = 
\begin{pmatrix} 
0 & \langle g| \\ g(\omega) & \omega 
\end{pmatrix}
\begin{pmatrix}
x \\ \xi(\omega)
 \end{pmatrix} 
 =
\begin{pmatrix}   
\langle g | \xi \rangle  \\ x g(\omega)+ \omega \xi(\omega) 
\end{pmatrix},
\end{equation}
on the domain  $\D (H_g) = \C \oplus \D(q)$, where 
$\D(q)= \{\xi \in L^2(\R), q \xi \in L^2(\R) \}$ is the domain of the position operator, $(q\xi) (\omega) = \omega \xi(\omega)$.

This Hamiltonian is the restriction to the vacuum/one-particle  sector 
of the  quantum-field Hamiltonian
\begin{equation}
\label{eq:H0F}
H_{\mathrm{L}} =   \int_{\mathbb{R}}\d\omega\, \omega\, a^*_\omega a_\omega + \int_{\mathbb{R}} \d\omega  \big(g(\omega) a_\omega^* +  \overline{g(\omega)} a_\omega\big)
\end{equation}
on the symmetric Fock space $\mathcal{F}_s(L^2(\mathbb{R}))$ where $a_\omega$ and $a_{\omega}^*$ are the bosonic annihilation and creation operators~\cite{Reed}.
Indeed, $\mathcal{F}_{n\leq1}
=\mathbb{C}\oplus L^2(\mathbb{R})=\mathcal{H}_e$ and, by noting that $|\mathrm{vac}\rangle = (1,0)$ and $a^*(\xi) |\mathrm{vac}\rangle = (0,\xi)$, the restriction of~(\ref{eq:H0F}) to $\mathcal{H}_e$ gives~(\ref{eq:matrixaction}). This model, introduced by Lee~\cite{Lee} as a solvable quantum-field model for studying the renormalisation problem, describes the decay of an unstable vacuum into the one-particle sector, due to an interaction term with coupling function $g$. When the coupling becomes flat, on physical ground one expects that the decay will be purely exponential. 

While it is tempting so simply put $g=\operatorname{constant}$ in Eq.~(\ref{eq:H0F}), this does not result in a self-adjoint Hamiltonian \cite{Wal}. Instead, one can show that, given a uniformly bounded sequence of positive coupling functions $(g_n)_{n\in\N}\subset L^2(\R) \cap L^{\infty}(\R)$, with 
$$g_n(\omega) \to 1/\sqrt{2\pi}$$ 
pointwise as $n\to\infty$, there exists a self-adjoint operator $H_{\star}$ which is the limit
$$ H_{g_n} \to H_\star$$
in the strong-resolvent sense, and thus~\cite{deOliveira}
$$\rme^{\rmi t H_{g_n}} \ra  \rme^{\rmi t H_\star}, \quad n\ra\infty,$$
strongly, for each $t\in\R$. 

One can explicitly write the unitary time evolution under the limit Hamiltonian $H_\star$.
It is convenient to consider the Fourier transform on the second component of~(\ref{eq:vector}), which leads us  to
\begin{equation*}
\begin{pmatrix}
x\\ \xi
\end{pmatrix}
\mapsto
\rme^{-\rmi t H_\star}
\begin{pmatrix}
x\\ \xi
\end{pmatrix}
=
\begin{pmatrix}
\rme^{-t/2} x - \rmi \rme^{-t/2} \int_0^t \rme^{s/2} \xi(s)\rmd s \\
\xi(t+\cdot) -\rmi \chi_{[-t,0]}(\cdot) \rme^{-(t+\cdot)/2} x
-\chi_{[-t,0]}(\cdot) \int_0^{t+\cdot} \rme^{-(t+\cdot)/2}\rme^{s/2} \xi(s) \rmd s 
\end{pmatrix},
\end{equation*}
for $t\ge 0$. Here $\chi_\Omega$ is the characteristic function of set $\Omega\subset\R$, i.e., $\chi_\Omega(t)=1$ if $t\in\Omega$ and $=0$ otherwise.
In particular, the vacuum state $|\mathrm{vac}\rangle = (1,0)$ will exponentially decay into a one-photon state as
\begin{equation*}
\rme^{-\rmi t H_\star}
\begin{pmatrix}
1\\ 0 
\end{pmatrix}
=
\begin{pmatrix}
\rme^{-t/2}  \\
 -\rmi \chi_{[-t,0]}(\cdot) \rme^{-(t+\cdot)/2} x 
\end{pmatrix}.
\end{equation*}

This implies that the spectrum of $H_\star$ is the whole real line. The coupling between a single qubit and the singular Friedrichs-Lee model $H_\star$ provides a reasonably realistic model of a quantum system in interaction with a Markovian environment. It is the dilation of the amplitude damping Lindbladian discussed in~\cite[Sec.~3]{AHFB15}, and it exhibits exponential decay. 
The single qubit is our system $\H_s=\C^2$ introduced in Ex.~\ref{ex:qubit}. This way the total Hilbert space is then $\H= \H_s\t\H_e$. We can write this out as $\H = \C \oplus L^2(\R) \oplus \C\oplus L^2(\R)$. This way the free time evolution $U_t$, for $t\ge 0$, of the coupled total system can be defined as
\begin{equation*}
\begin{pmatrix}
x_1\\\xi_1\\ x_2\\ \xi_2
\end{pmatrix}
\mapsto
\begin{pmatrix}
\rme^{-t/2} x_1 - \rmi \rme^{-t/2} \int_0^t \rme^{s/2} \xi_2(s)\rmd s \\
\xi_1\\
x_2\\
\xi_2(t+\cdot) -\rmi \chi_{[-t,0]}(\cdot) \rme^{-(t+\cdot)/2}x_1
-\chi_{[-t,0]}(\cdot) \int_0^{t+\cdot} \rme^{-(t+\cdot)/2}\rme^{s/2} \xi_2(s) \rmd s
\end{pmatrix}.
\end{equation*}

From this time evolution, one could now compute the Hamiltonian, following the lines of~\cite[Sec.4.2.2]{Wal}, and show that the conditions in Thm.~\ref{th:suff} are not fulfilled. But since we want to show that decoupling does not work for this model, we proceed differently and compute the evolution explicitly.

Consider as initial state the vector $(1,0)\t (1,0)$, i.e.,
\[
\begin{pmatrix}
x_1\\\xi_1\\ x_2\\ \xi_2
\end{pmatrix}
=
\begin{pmatrix}
1\\ 0 \\ 0\\ 0
\end{pmatrix},
\]
we get the free time evolution from time $0$ to $t$ as
\[
\begin{pmatrix}
1 \\0\\0\\0
\end{pmatrix}
\mapsto
\begin{pmatrix}
\rme^{-t/2}\\
0\\
0\\
-\rmi \chi_{[-t,0]}(\cdot) \rme^{-(t+\cdot)/2}
\end{pmatrix}.
\]
This shows that the the state of the qubit system decays exponentially, due to interaction with the environment. 

We would now like to show that this still happens when dynamical decoupling is applied. We choose the group generated by the four Pauli matrices as decoupling set $V$. Following~\eqref{eq:decevo}, at time $4\tau$ we find the total perturbed time evolution
\begin{equation*}
\begin{pmatrix}
x_1\\ \\ 0 \\ \\ 0 \\ \\
 \xi_2 
\end{pmatrix}
\mapsto
\begin{pmatrix}
\rme^{-\tau} x_1 - \rmi \rme^{-\tau} \int_0^{\tau} \rme^{s/2} \xi_2(s)\rmd s + \rmi \rme^{-\tau} \int_{\tau}^{2\tau} \rme^{s/2} \xi_2(s)\rmd s \\ \\
0 \\ \\
0 \\\\
\xi_2(2\tau+\cdot) -\rmi \chi_{[-2\tau,-\tau]}(\cdot) \rme^{-(2\tau+\cdot)/2}x_1
+\rmi \chi_{[-\tau,0]}(\cdot) \rme^{-(2\tau+\cdot)/2}x_1\\
-\chi_{[-2\tau,-\tau]}(\cdot) \int_0^{2\tau+\cdot} \rme^{-(2\tau+\cdot)/2}\rme^{s/2} \xi_2(s) \rmd s\\
+\chi_{[-\tau,0]}(\cdot) \int_0^{\tau}x_1 \rme^{-(2\tau+\cdot)/2}\rme^{s/2} \xi_2(s) \rmd s
-\chi_{[-\tau,0]}(\cdot) \int_{\tau}^{2\tau+\cdot} \rme^{-(2\tau+\cdot)/2}\rme^{s/2} \xi_2(s) \rmd s
\end{pmatrix}.
\end{equation*}
Iterating this cycle now $n$ times, we can compute the perturbed time evolution at time $t=4\tau n$ from this formula. However, we do not require such a level of generality since we are mainly interested in the evolution of the state $(1,0,0,0)$. A closer inspection under the assumption that the initial state is $(1,0,0,0)$ shows the following: the second and third component remain $0$ at time $t$; the fourth component evolves to a function with support on $[-2\tau n,0]$, so on the negative half-axis, so that the term
\[
- \rmi \rme^{-\tau} \int_0^{\tau} \rme^{s/2} \xi_2(s)\rmd s + \rmi \rme^{-\tau} \int_{\tau}^{2\tau} \rme^{s/2} \xi_2(s)\rmd s
\]
vanishes. Therefore,  we get
\begin{equation*}
\begin{pmatrix}
1 \\0\\0\\0
\end{pmatrix}
\mapsto
\begin{pmatrix}
\rme^{-t/4}\\
0\\
0\\
\rmi \phi_{n,t}
\end{pmatrix},
\end{equation*}
where
\begin{equation}
\phi_{n,t}(s) =  \rme^{-(t/2+s)} \sum_{k=1}^n \Big(
\chi_{[-tk/2n,-tk/2n+t/4n]}(s) - \chi_{[-tk/2n+t/4n,-t(k-1)/4n]}(s) \Big),
\label{eq:LFnth}
\end{equation}
so we still get exponential decay on the system if the initial state was $(1,0,0,0)$. We are interested in the limit $n\ra\infty$ with $\tau\ra 0$ such that $t=4\tau n$ remains fixed. Then we should get
\begin{equation}
\label{eq:LFlimit}
\begin{pmatrix}
1\\0\\ 0\\ 0
\end{pmatrix}
\mapsto
\begin{pmatrix}
\rme^{-t/4}\\ 0 \\ 0 \\ \phi_t 
\end{pmatrix},
\end{equation}
with some function $s\mapsto\phi_t(s)$, but it turns out impossible to obtain the limit $\phi_t$ because~\eqref{eq:LFnth} does not converge as $n\ra \infty$. See Fig.~\ref{fig:LFenvironmentDD}.

Notice, however, that $\phi_{n,t}$ converges weakly to zero, that is $\langle f | \phi_{n,t}\rangle \to 0$ as $n\to\infty$ for all $f\in L^2(\R)$. Physically, one can interpret the behaviour of  the wave function $\phi_{n,t}$ as the result of   pumping  larger and larger energy in the system through the decoupling pulses. In the limit $n\to\infty$ the pumped energy becomes infinite and $\phi_{n,t}$ gets orthogonal to any given wave function.

\begin{figure}
\centering{}\includegraphics[width=0.8\columnwidth]{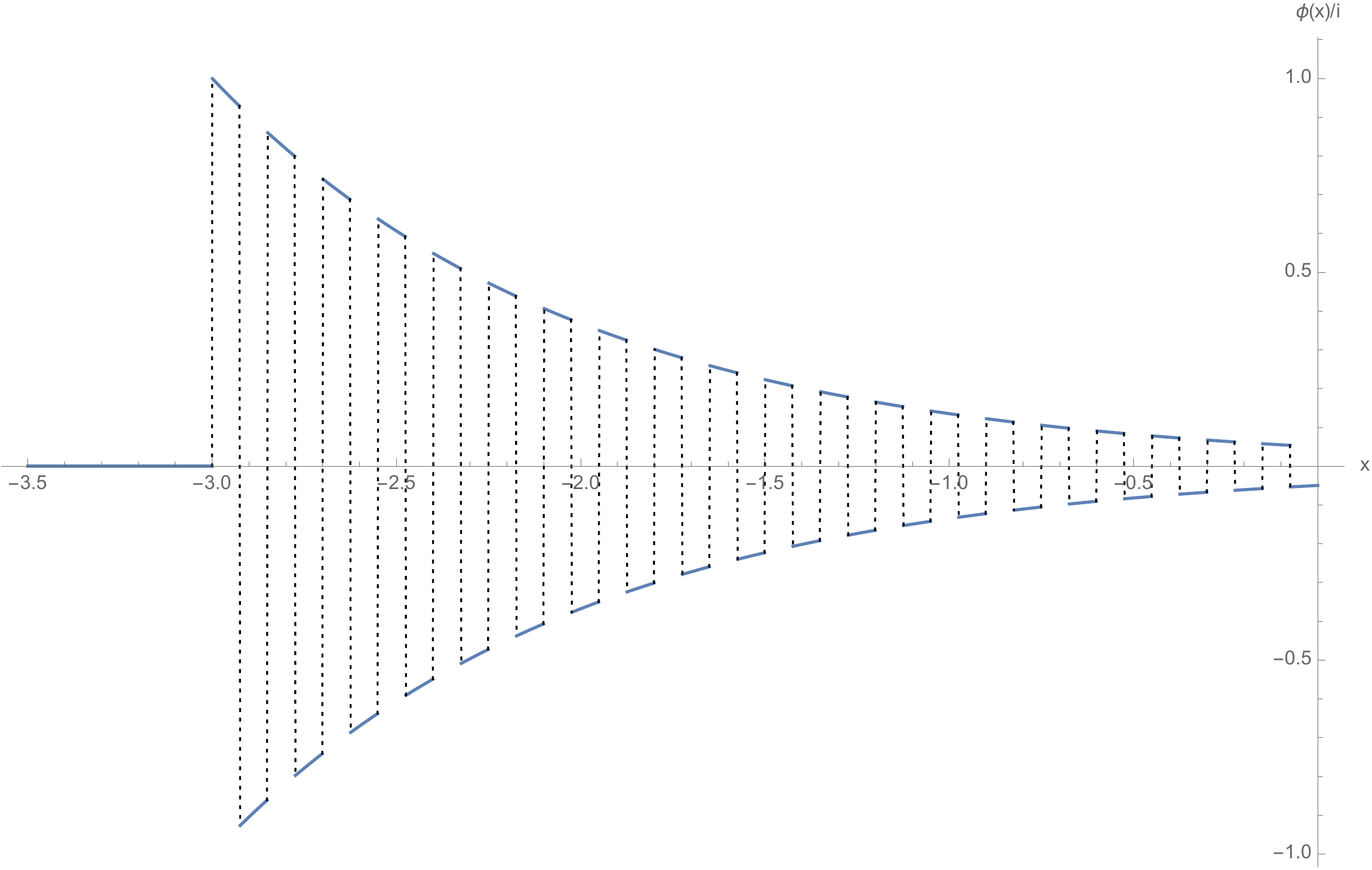}\caption{\label{fig:LFenvironmentDD}
Plot of the function $\phi_{n,t}(s)$ versus $s$. We set $t=6$ and $n=20$. 
}
\end{figure}

In any case, if dynamical decoupling worked then there would be a self-adjoint $B$ such that the time evolution of $(1,0,0,0)=(1,0)\t(1,0)$ is given by
\[
\begin{pmatrix}
1\\0
\end{pmatrix}
\t
\begin{pmatrix}
1\\0
\end{pmatrix}
\mapsto 
\begin{pmatrix}
1\\0
\end{pmatrix}
\t
\rme^{\rmi t B}
\begin{pmatrix}
1\\0
\end{pmatrix} 
=
\begin{pmatrix}
1\\0
\end{pmatrix}
\t
v(t)
\]
with some $v(t)\in\H_e$ of norm $1$; this in turn would be a vector of the form
\[
\begin{pmatrix}
* \\ * \\0\\0
\end{pmatrix}.
\]
Given that the first component was $\rme^{-t/4}$, we see that the second should be nonzero, which contradicts~\eqref{eq:LFlimit} where the second component is $0$. Therefore there cannot be such $B$ and thus dynamical decoupling does not work for this model.

\section{Conclusions}
We have provided  criteria and examples for dynamical decoupling of unbounded Hamiltonians. From a mathematical perspective, the ability to decouple is essentially a question of approximate commutativity and of operator domains. From a physical perspective, it is a question of interaction time-scales, but we saw that such time-scales cannot be revealed by looking at the reduced dynamics only. Moreover, even if the complete model is known and can be decoupled, such time-scales are very hard to compute in practice, because they depend explicitly on the environment initial state.

In practice, our results imply that many more systems can be protected from environmental noise than previously thought. In particular, seeing exponential decay in the lab should not stop one from trying to apply decoupling.   Whether or not it works on a feasible time-scale can be decided experimentally.

This paper is probably the first to discuss dynamical decoupling of unbounded Hamiltonians in a systematic way. We saw that the question of whether decoupling works is a very hard and delicate one and we believe a precise characterization, e.g.~an affirmative answer to Problem \ref{prob2}, is still a long way off; yet we provided some new and very useful methods to start with.

Apart from the physical motivation, our work also offers a refreshed view on the mathematics of Trotter product limits. On the one hand, established results on convergence can be embedded into a dynamical decoupling model through the construction (\ref{ex:sp}), and get a physical meaning. On the other hand, a proof of long-standing conjectures such as the generalization of \cite{Kato} to more than two generators would be highly relevant in dynamical decoupling.

\section*{Acknowledgments}

We acknowledge discussions with Martin Fraas, Jukka Kiukas, Marilena Ligab\`o, Alessio Serafini, John Gough and Benjamin Dive. PF thanks Kenji Yajima for some valuable remarks and for pointing out Chernoff's example discussed in Remark~\ref{rm:cher}. 
DB acknowledges support by the EPSRC Grant No. EP/M01634X/1.
PF was partially supported by INFN through the project ``QUANTUM'', and by the Italian National Group of Mathematical Physics (GNFM-INdAM).


\begin{thebibliography}{1}




\bibitem{ABH17} C. Arenz, D. Burgarth, R. Hillier. J. Phys. A: Math. Theor. \textbf{50}, 135303 (2017).

\bibitem{AHFB15} C. Arenz, R. Hillier, M. Fraas, D. Burgarth. Phys.
Rev. A \textbf{92}, 022102 (2015).

\bibitem{Che} P. Chernoff. \textit{Product formulas, nonlinear semigroups and addition of unbounded operators}. Memoirs AMS 140 (1974).

\bibitem{deOliveira}
C.R. de Olivera.  \textit{Intermediate spectral theory and quantum dynamics}. Progress in Mathematical Physics Vol. \textbf{54}. Birkh\"auser (2009).

\bibitem{Exner}P. Exner, H. Neidhardt and V. A. Zagrebnov. Integr.
Equ. Oper. Theory \textbf{69}, 451 (2011).

\bibitem{symplectic}
A. Ferraro S. Olivares, M.G.A. Paris. \textit{Gaussian states in quantum information}.
Napoli Series on physics and Astrophysics.
Bibliopolis, Napoli  (2005).

\bibitem{Friedrichs} K.O. Friedrichs. Comm. Pure Appl. Math. \textbf{1},
361 (1948).

\bibitem{GN17} J. Gough and H. Nurdin. Private communication.

 \bibitem{HW68} U. Haeberlen and J. S. Waugh. Phys. Rev. \textbf{175}, 453 (1968). 

\bibitem{ABH15} R. Hillier, C. Arenz, D. Burgarth. J. Phys. A: Math. Theor. \textbf{48}, 155301

\bibitem{HOT} D.T.  Hoa, H. Osaka, J.Tomiyama. Linear and Multilinear Algebra \textbf{63}, 1577 (2015).

\bibitem{JL} G.W. Johnson, M.L. Lapidus. \textit{The Feynman integral and Feynman's operation calculus}. Oxford University Press (2002).

\bibitem{Kato} T. Kato. In \textit{Topics in functional analysis}, Adv. Math. Supplementary Studies \textbf{3}, 185--195 (1978).

\bibitem{KM} T. Kato and K. Masuda. J. Math. Soc. Japan \textbf{30}, 169-178 (1978).

\bibitem{KL08} K. Kaveh and D. A. Lidar. Phys. Rev. A \textbf{78}, 012355 (2008).

\bibitem{Lee} T.D. Lee. Phys. Rev. \textbf{95}, 1329 (1956).

\bibitem{LB13} D. A. Lidar and T. A. Brun.
\textit{Quantum error correction}, Cambridge University Press, Cambridge
(2013).

\bibitem{Nel} E. Nelson. Ann. Math. \textbf{70}, 572-615 (1959).

\bibitem{Reed} M. Reed and B. Simon.
\textit{Functional analysis}, Academic Press, London, (1980).
 
 \bibitem{VKL99} L. Viola, E. Knill and S. Lloyd. Phys. Rev. Lett. \textbf{82}, 2417 (1999).

\bibitem{VK05} L. Viola and E. Knill. Phys. Rev. Lett. \textbf{94}, 060502 (2005).
 
 \bibitem{VKL98} L. Viola and S. Lloyd. Phys. Rev. A \textbf{58}, 2733 (1998).

\bibitem{Wal} W. von Waldenfels. \textit{A measure theoretical approach to quantum stochastic processes}. Springer, Berlin (2014).

 \bibitem{WHH68} J.S. Waugh, L. M. Huber and U. Haeberlen. Phys. Rev. Lett. \textbf{20}, 180 (1968). 





\end{thebibliography}
\end{document}